\documentclass[letterpaper, 10 pt, journal, final]{IEEEtran}  

\IEEEoverridecommandlockouts 

\usepackage[numbers]{natbib}
\usepackage{multicol}
\usepackage[bookmarks=true,hidelinks]{hyperref}
\usepackage{url}
\usepackage{graphics} 
\usepackage{times} 
\usepackage{amsmath} 
\usepackage{amssymb} 
\usepackage{xcolor}
\usepackage{graphicx}
\usepackage{amsfonts}
\usepackage{mathtools}
\usepackage{caption}
\usepackage[T1]{fontenc} 
\usepackage[hyphenbreaks]{breakurl} 
\usepackage{mathtools, cuted} 
\usepackage[ruled,vlined]{algorithm2e} 
\usepackage{enumerate}
\usepackage{tikz}
\usepackage{comment}

\usepackage{subcaption}
\DeclareCaptionLabelSeparator{periodspace}{.\quad}
\usepackage{amsthm}

\newtheorem{lemma}{Lemma}
\newtheorem{remark}{Remark}
\newtheorem{proposition}{Proposition}
\newtheorem{theorem}{Theorem}

\theoremstyle{definition}

\newtheorem{defn}{Definition}

\begin{document}

\title{\LARGE \bf
Maneuvering and robustness issues in undirected displacement-consensus-based formation control
}
\author{Hector Garcia de Marina%
	\thanks{H. G. de Marina is with the Department of Computer Architecture and Automatic Control at the Faculty of Physics, Universidad Complutense de Madrid, ~e-mail: hgarciad@ucm.es . This work has been supported by the grant \emph{Atraccion de Talento} 2019-T2/TIC-13503 from the Government of the Autonomous Community of Madrid, and by the Spanish Ministry of Science and Innovation under research Grant RTI2018-098962-B-C21.}
}


\maketitle

\begin{abstract}
	In this paper, we first propose a novel maneuvering technique compatible with displacement-consensus-based formation controllers. We show that the formation can be translated with an arbitrary velocity by modifying the weights in the consensus Laplacian matrix. In fact, we demonstrate that the displacement-consensus-based formation control is a particular case of our more general method. We then uncover robustness issues with undesired steady-state motions and resultant distorted shapes in undirected displacement-consensus-based formation control. In particular, these issues are triggered when neighboring agents mismeasure their relative positions, e.g., their onboard sensors are misaligned and have different scale factors. We will show that if all the sensing is close to perfect but different among the agents, then the stability of the system is compromised. Explicit expressions for the eventual non-desired velocity and shape's distortion are given as functions of the scale factors and misalignments for formations based on tree graphs.
\end{abstract}


\IEEEpeerreviewmaketitle


%



\section{Introduction}
Distributed robot swarms can be more effective, flexible, fault-tolerant, and scalable than the traditional monolithic task-specific robot. A distributed robot swarm creates global behaviors emerging from the local interactions between its members. For example, a swarm can display a particular geometric shape by having its individuals controlling their geometric relations in between such as relative positions, relative angles, or distances. However, the scientific community is still on the development of reliable and systematic methods for the control of robot swarms \cite{yang2018grand}. In particular, a series of robustness issues show up when robots have a different perception than their neighbors while they run distributed formation controllers. For example, it was reported that biased range sensors in distance-based formation controllers designed from undirected graphs cause instabilities in the sense of undesired steady-state motions \cite{mou2016undirected,Marina15}. In undirected displacement-based formation control, if agents do not share the same reference for their orientation in 2D, then again an undesired motion shows up for the resulting distorted formation \cite{meng2016formation}. Following the same pattern, it was recently reported that when agents have biased inter-agent angle measurements, again undesired steady-state collective motions are present \cite{chen2019triangular}. Although it has not been generalized yet for arbitrary distributed formation controllers, the robots can correct their mismatched sensors with respect to their neighbors and solve the mentioned issues for some particular formation control strategies by using adaptive controllers \cite{meng2016formation,Marina2017Taming}. Conversely, these robustness issues can be seen as an opportunity to maneuver the whole formation by, for example, injecting sensor biases in purpose to induce (desired) steady-state motions \cite{Marina16,Marina2018stability}.

In this paper, we will focus on displacement-consensus-based formation control. According to the literature \cite{Oh2015}, the term \emph{displacement} refers to neighboring robots (or agents in general) that control their relative positions by only measuring and comparing them with a common relative target vector. In order to have a successful outcome, all the agents must represent the target vectors with respect to the same frame of coordinates. Therefore, it is a requirement that neighboring agents have a common frame of coordinates as a reference, or at least, they must know each other's local reference frames. The term \emph{consensus} appears because if the set of desired relative positions is a collection of zero-valued vectors, then the formation control strategy becomes the standard consensus algorithm by using the Laplacian matrix.

There are two main contributions in this paper. Firstly, we show that by manipulating the weights of the \emph{standard} Laplacian matrix, the swarm can display the target shape plus a desired steady-state translational motion. Similarly as in \cite{Marina2017Taming,Marina2018stability}, where mismatches are introduced in target distances to induce collective motions, in this paper, we introduce mismatches in the weights coming from the Laplacian for the same purpose. We will show that the displacement-consensus-based formation control is a particular case of our proposed methodology. Secondly, we show how different scale factors and misalignments in onboard sensing for measuring relative positions result in an undesired traveling distorted (with respect to the target) shape. If such a mismeasure is close to perfect and equal for all the agents, the resultant shape is distorted but there is no steady-state motion. On the other hand, if the mismeasure is far from perfect, regardless of being equal for all the agents, then in general, the formation is unstable in the sense of having the global (and relative) positions of the agents growing unbounded exponentially fast.

We want to remark that the problem of analyzing the presence of \emph{mismatched compasses} in 2D displacement-based formation control has been discussed in \cite{meng2016formation, ahn2019consensus}. We also cover such a problem in the second half of our paper. However, not only we combine it with the presence of different scale factors, but we provide explicit expressions for the distorted shape and the undesired velocity. Furthermore, the presented analysis and expressions are valid for the $m$-dimensional case.

\section{Preliminaries}
\label{sec: pre}
\subsection{Notation}
In this paper, we will focus on formations of mobile agents in $m\in\mathbb{N}$ dimensions. Given a matrix $A\in\mathbb{R}^{p\times q}$, we define the operator $\overline A := A \otimes I_m \in \mathbb{R}^{pm \times qm}$, where $\otimes$ denotes the Kronecker product. Given a stacked vector $x := \begin{bmatrix}x_1^T & x_2^T & \dots & x_k^T\end{bmatrix}^T$ with $x_i\in\mathbb{R}^m, i\in\{1,\dots,k\}$, we define the operator $D_x := \operatorname{diag}\{x_i\}_{i\in\{1,\dots,k\}}\in\mathbb{R}^{k\times km}$, and $||x||$ denotes its Euclidean norm. Given a set $\mathcal{X}$, we denote by $|\mathcal{X}|$ its cardinality. We denote by $\mathbf{1}_p\in\mathbb{R}^p$ the all-one column vector, and finally, we also denote by $\mathbf{\hat 1}_l\in\mathbb{R}^m, 1\leq l \leq m$ the column vector with $m$ components whose $l$'th position is equal to one and the rest are zero, i.e., one of the elements of the standard basis for the $m$-dimensional Euclidean space.

\subsection{Graph theory}
\label{sec: not}
A \emph{graph} $\mathcal{G} = (\mathcal{V}, \mathcal{E})$ consists of two non-empty sets: the node set $\mathcal{V} = \{1,\dots,n\}$ with $n \geq 2$, and the ordered edge set $\mathcal{E} \subseteq (\mathcal{V}\times\mathcal{V})$. For an arbitrary edge $\mathcal{E}_k = (\mathcal{E}_k^{\text{head}},\mathcal{E}_k^{\text{tail}})$, we call to its first and second element the \emph{tail} and the \emph{head} respectively. The set $\mathcal{N}_i$ containing the neighbors of the node $i$ is defined by $\mathcal{N}_i:=\{j\in\mathcal{V}:(i,j)\in\mathcal{E}\}$. Let $w_{ij}\in\mathbb{R} \neq 0$ a weight associated with the edge $\mathcal{E}_k=(i,j), k\in\{1,\dots,|\mathcal{E}|\}$, then the \emph{Laplacian} matrix $L\in\mathbb{R}^{n\times n}$ of $\mathcal{G}$ is defined as
\begin{equation}
	l_{ij} := \begin{cases}\sum_{k\in\mathcal{N}_i}w_{ik} & \text{if} \quad i = j \\
		-w_{ij} & \text{if} \quad i \neq j \wedge j\in\mathcal{N}_i \\
		0 & \text{if} \quad i \neq j \wedge j\notin\mathcal{N}_i.
	\end{cases}
	\label{eq: L}
\end{equation}
In this paper we deal with the special case of \emph{undirected} graphs. In particular, undirected graphs are \emph{bidirectional} graphs where each edge $\mathcal{E}_k$ is transformed into two directed edges $(i,j)$ and $(j,i)$. For an undirected graph, we choose one of the two arbitrary directions for each $\mathcal{E}_k$ and we construct the following \emph{incidence matrix} $B\in\mathbb{R}^{|\mathcal{V}|\times |\mathcal{E}|}$ of $\mathcal{G}$
\begin{equation}
	b_{ik} := \begin{cases}+1 \quad \text{if} \quad i = {\mathcal{E}_k^{\text{tail}}} \\
		-1 \quad \text{if} \quad i = {\mathcal{E}_k^{\text{head}}} \\
		0 \quad \text{otherwise.}
	\end{cases}
	\label{eq: B}
\end{equation}
For an undirected graph, if $\omega_{ij} = \omega_{ji} = \omega_k$ with $k$ corresponding to the edge $\mathcal{E}_k$, and we stack all the edges' weights $w_k$ in $w\in\mathbb{R}^{|\mathcal{E}|}$, then it can be checked the following relation
\begin{equation}
L = BD_wB^T.
	\label{eq: Lwk}
\end{equation}
If the graph $\mathcal{G}$ is \emph{connected}, then the Laplacian matrix $L$ has a single eigenvalue equal to zero, whose associated eigenvector is $\mathbf{1}_n$. Note that if $w_k > 0, \forall k$, then $L$ is positive semidefinite.

\section{Displacement-consensus-based formation control}
\subsection{Frameworks and desired shapes}
We consider a team consisting of $n\geq 2$ agents where each agent $i$ has a position $p_i\in\mathbb{R}^m$. We stack all the positions $p_i$ in a single vector $p\in\mathbb{R}^{mn}$ and we call it \emph{configuration}. We define a \emph{framework} $\mathcal{F}$ as the pair $(\mathcal{G}, p)$, where we assign each agent's position $p_i$ to the node $i\in\mathcal{V}$, and the graph $\mathcal{G}$ establishes the set of neighbors $\mathcal{N}_i$ for each agent $i$.

We choose an arbitrary configuration of interest $p^*$ for the team of agents, and we split it as
\begin{equation}
	p^* = \left(\mathbf{1}_n \otimes p_{\text{c.m.}}\right) + p^*_c,
	\label{eq: pstar}
\end{equation}
where $p_{\text{c.m.}}\in\mathbb{R}^{m}$ is the position of the \emph{center of mass} of the configuration and $p_c^*\in\mathbb{R}^{mn}$, starting from the center of mass, gives the \emph{appearance} to the formation as in the example shown in Figure \ref{fig: pstar}. Without loss of generality, and for the sake of simplicity, we set $p_{\text{c.m.}} = 0$ in (\ref{eq: pstar}), i.e., $p^* = p^*_c$.
\begin{figure}
\centering
\begin{tikzpicture}[line join=round,scale=1]
\filldraw(0,0) circle (2pt);
\filldraw(1.5,0) circle (2pt);
\filldraw(1.5,1.5) circle (2pt);
\filldraw(0,1.5) circle (2pt);
\draw[draw=black,arrows=->](-2,0)--(-1.5,0);
\draw[draw=black,arrows=->](-2,0)--(-2,0.5);
\draw[draw=black,arrows=->](-2,0)--(0.75,0.75);
\draw[draw=black,arrows=->](0.75,0.75)--(0+0.05,1.5-0.05);
\draw[draw=black,arrows=->](0.75,0.75)--(0+0.05,0+0.05);
\draw[draw=black,arrows=->](0.75,0.75)--(1.5-0.05,1.5-0.05);
\draw[draw=black,arrows=->](0.75,0.75)--(1.5-0.05,0+0.05);
\node at (-1.95,-0.25) {\small $O_g$ \normalsize};
\node at (0-0.2,1.2) {\small $p_{c_1}^*$ \normalsize}; 
\node at (0-0.2,0.75-0.45) {\small $p_{c_3}^*$ \normalsize}; 
\node at (1.5+0.2,0.75+0.45) {\small $p_{c_2}^*$ \normalsize}; 
\node at (1.5+0.2,0.75-0.4) {\small $p_{c_4}^*$ \normalsize}; 
\node at (-0.5,0.65) {\small $p_{\text{c.m.}}$ \normalsize}; 
\end{tikzpicture}
	\caption{Example of a particular 2D configuation $p^*$ that can be constructed by $p^* =  (\mathbf{1}_4 \otimes p_{\text{c.m.}}) + [p_{c_1}^{*T} \, p_{c_2}^{*T} \, p_{c_3}^{*T} \, p_{c_4}^{*T}]^T$, where $p_{\text{c.m.}}$ is the center of mass of the desired formation.}
\label{fig: pstar}
\end{figure}
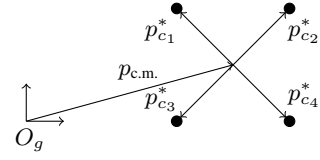

We now define the concept of \emph{desired shape} constructed from the configuration of interest or \emph{reference shape} $p^*$:
\begin{defn}
	The framework (or formation) is at the \emph{desired shape} when 
	\begin{equation}
	p \in \mathcal{S} := \{p \, : \, p = p^* + (\mathbf{1}_n \otimes b), \, b\in\mathbb{R}^m \}.
	\label{eq: dshape}
	\end{equation}
\end{defn}
Then, we can deduce that only translations of $p^*$ are admissible as desired shape.  

\subsection{Agents' dynamics and shape stabilization}
\label{sec: dyn}
In this paper we consider that the position's dynamics of each agent $i\in\mathcal{V}$ are modelled by the following single-integrator
\begin{equation}
	\dot p_i = u_i,
	\label{eq: dyn}
\end{equation}
where $u_i\in\mathbb{R}^m$ is the control action for the corresponding agent $i$. Since the displacement-consensus-based formation control is distributed, then the agent $i$ only has access to relative information with respect to its neighbors in $\mathcal{N}_i$. This requirement implies that the maneuvering technique to be introduced in this paper can only count on the same available information as well. In particular, such a local available information is the set of relative positions $z_{ij} = p_i - p_j, (i,j)\in\mathcal{E}$. We can calculate the stacked vector of sensed relative positions using the incidence matrix as follows
\begin{equation}
z = \overline B^Tp.
	\label{eq: z}
\end{equation}

In order to have a distributed control action $u_i$, it must be of the form
\begin{equation}
	u_i := f(z_{ij}), \quad j\in\mathcal{N}_i
	\label{eq: fij}
\end{equation}
where we set $f : \mathbb{R}^{m|\mathcal{N}_i|} \rightarrow \mathbb{R}^m$ to be linear. We then combine (\ref{eq: fij}) and (\ref{eq: dyn}) in the following compact form
\begin{equation}
\dot p = u,
	\label{eq: pdyn}
\end{equation}
where $u\in\mathbb{R}^{mn}$ is the stacked vector of control actions $u_i$.

The displacement-consensus-based controller (\ref{eq: fij}) for the dynamics (\ref{eq: dyn}) is given by \cite{Oh2015}
\begin{align}
	u_i &= - \sum_{j\in\mathcal{N}_i} \omega_{ij} \left(p_i - p_j - (p^*_i - p^*_j)\right) \nonumber \\ &= -\sum_{j\in\mathcal{N}_i} \omega_{ij} \left(z_{ij} -z^*_{ij}\right),
	\label{eq: uLcon}
\end{align}
which can be written in compact form for (\ref{eq: pdyn}) as
\begin{equation}
	u = -\overline L(p - p^*),
	\label{eq: Lcon}
\end{equation}
where $L$ is as in (\ref{eq: Lwk}), i.e., $\omega_{ij}=\omega_{ji}=\omega_k$, we set $\omega_k > 0$, and if $p^* = 0$, then (\ref{eq: Lcon}) becomes the standard consensus algorithm. Following the protocol (\ref{eq: Lcon}), if $\mathcal{G}$ in the framework is connected, then we have that $p(t) \to p^* + \sum_{l=1}^m c_l(\mathbf{1}_n \otimes \mathbf{\hat 1}_l)$ as $t\to\infty$, where each $c_l\in\mathbb{R}$ is determined by the initial condition $p(0)$. We note that $p(t)$ converges to a point where the formation stays at the desired shape $\mathcal{S}$ and there is no further stationary motion or maneuvering.

\section{Shape maneuvering for displacement-consensus-based formation control}
\label{sec: displacement}
The maneuvering strategy in this paper consists in modifying the control (\ref{eq: Lcon}) as
\begin{equation}
u = -\overline L_m p + \overline L p^*,
\end{equation}
where $L_m\in\mathbb{R}^{n\times n}$ is a \emph{modified Laplacian} matrix with modified weights from the original $L$. Obviously, if $L_m = L$, we recover (\ref{eq: Lcon}). This modification will allow the formation to converge to the desired shape defined from $p^*$ and to a desired translational motion. In particular, the modification of the weights $\omega_{ij}$ will be done by exploiting the relative position vectors between the agents in the reference shape $p^*$.

\subsection{Modified Laplacian matrix for motion control}
Let us consider the following weights for constructing a modified Laplacian matrix
\begin{equation}
\tilde\omega_{ij} = \omega_{ij} - \kappa\mu_{ij},
\label{eq: wmod}
\end{equation}
where $\omega_{ij}$ are the weights of the original Laplacian matrix (\ref{eq: L}), $\kappa\in\mathbb{R}$, $\mu_{ij}\in\mathbb{R}$, and if $j\notin\mathcal{N}_i$, then $\mu_{ij} = 0$. We anticipate that the second term on the right-hand side of (\ref{eq: wmod}) is responsible for the steady-state motion of agent $i$. As we will see, $\mu_{ij} \neq \mu_{ji}$ in general. Therefore, $\tilde\omega_{ij} \neq \tilde\omega_{ji}$ for the modified Laplacian matrix.

In a first approach, we design the desired steady-state velocity $v^*_i\in\mathbb{R}^m$ for each agent $i\in\mathcal{V}$ as linear combinations of the desired relative positions $(p^*_i - p^*_j), (i,j)\in\mathcal{E}$. Formally, the steady-state velocity for agent $i$ can be designed by finding a set of $\mu_{ij}$'s that satisfies
\begin{equation}
	v^*_i = \kappa \sum_{j\in\mathcal{N}_i} \mu_{ij} (p^*_i - p^*_j).
	\label{eq: vs}
\end{equation}
If we can guarantee the convergence of the agents to such desired relative positions, we will see that, consequently, the agents will converge to the desired steady-state velocity as well. Note that the parameters $\mu_{ij}$ define the direction of $v^*_i$ and $\kappa$ regulates its speed. We remind that only the translation of $p^*$ is allowed so that the agents stay at the desired shape $\mathcal{S}$. This fact implies that $v^*_i = v^* \in\mathbb{R}^m, \forall i\in\mathcal{V}$, as illustrated in the example in Figure \ref{fig: sqdis}. 
\begin{remark}
In some particular cases, it is not possible to construct an arbitrary $v^*$ by only following (\ref{eq: vs}), e.g., agent $i$ has only one relative position available like agent $4$ in Figure \ref{fig: sqdis}. In fact, in order to have a solution for an arbitrary $v_i^*$ in (\ref{eq: vs}), a necessary condition for agent $i$ is to have at least $m$ neighbors. Nevertheless, to overcome such an issue, at the end of this section we will provide a more general approach for the design of the motion parameters $\mu_{ij}$ as square matrices instead of real numbers. However, for the sake of clarity in the notation and without loss of generality, we proceed to present the analysis with the motion parameters $\mu_{ij}$ as in (\ref{eq: wmod}) and (\ref{eq: vs}), and later extend the results.
\end{remark}

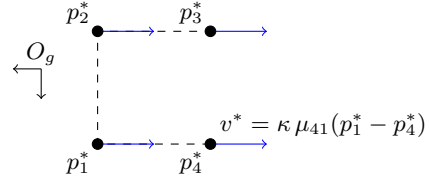
\begin{figure}
\centering
	\begin{tikzpicture}[line join=round,scale=1]
\draw[dashed](1.5,0)--(0,0)--(0,1.5)--(1.5,1.5);
\draw[draw=blue,arrows=->](0,0)--(.75,0);
\draw[draw=blue,arrows=->](1.5,0)--(2.25,0);
\draw[draw=blue,arrows=->](1.5,1.5)--(2.25,1.5);
\draw[draw=blue,arrows=->](0,1.5)--(.75,1.5);
\filldraw(0,0) circle (2pt);
\filldraw(1.5,0) circle (2pt);
\filldraw(1.5,1.5) circle (2pt);
\filldraw(0,1.5) circle (2pt);
\draw[draw=black,arrows=->](-.75,1.0)--(-1.125,1.0);
\draw[draw=black,arrows=->](-.75,1.0)--(-.75,0.6);
		\node at (-.75,1.2) {\small $O_g$ \normalsize};\node at (-.25,-.25) {\small $p_1^*$ \normalsize}; \node at (1.25,1.75) {\small $p_3^*$ \normalsize};\node at (-.25,1.75) {\small $p_2^*$ \normalsize}; \node at (1.25,-.25) {\small $p_4^*$ \normalsize}; \node at (3,.25) {\small $v^* = \kappa\,\mu_{41}(p_1^* - p_4^*)$ \normalsize};
\end{tikzpicture}
	\caption{Four agents displaying a \emph{square} reference shape $p^*$ with $\mathcal{E} = \{(1,2), (2,3), (1,4)\}$. The proposed algorithm makes the \emph{formation's velocity} in closed-loop with the relative positions $z_{ij}(t)$ by modifying the weights in the Laplacian as in (\ref{eq: wmod}) with the motion parameters $\mu_{ij}$. Therefore, once the desired shape $\mathcal{S}$ is achieved, we also achieve the desired velocity $v^*$. In order to stay eventually in $\mathcal{S}$, all the eventual velocities (blue vectors) must be designed equally in (\ref{eq: vs}).}
	\label{fig: sqdis}
\end{figure}

Let us define the components of the following matrix $M\in\mathbb{R}^{|\mathcal{V}|\times |\mathcal{E}|}$
\begin{equation}
	m_{ik} := \begin{cases}\mu_{i\mathcal{E}_k^{\text{head}}} \quad \text{if} \quad i = {\mathcal{E}_k^{\text{tail}}} \\
		-\mu_{i\mathcal{E}_k^{\text{tail}}} \quad \text{if} \quad i = {\mathcal{E}_k^{\text{head}}} \\
		0 \quad \text{otherwise.}
	\end{cases}.
	\label{eq: M}
\end{equation}
Based on the obvious identity $(p_i - p_j) = - (p_j - p_i)$, we can stack all the velocities (\ref{eq: vs}) from all the agents in a compact form as follows
\begin{equation}
\mathbf{1}_n \otimes v^*=\kappa \overline{MB^T}p^* = \kappa \overline\Lambda p^*.\label{eq: vsm}
\end{equation}
Note that we can decouple the different components of $v^*$ and $p^*$ in (\ref{eq: vsm}). For example, for the \emph{x-components} we have the following compact form
\begin{equation}
\mathbf{1}_n \otimes ({^xv}^*) = \kappa MB^T({^xp}^*) = \kappa \Lambda ({^xp}^*).
\label{eq: vs2}
\end{equation}
\begin{lemma}
\label{lem: lmod}
Consider the Laplacian matrix $L$ constructed from an undirected connected graph where $\omega_{ij}=\omega_{ji} > 0$. Also consider $\Lambda$ as in (\ref{eq: vs2}). Then, for a sufficiently small $|\kappa|$, the modified Laplacian matrix $(L-\kappa\Lambda)$ has a single zero eigenvalue whose eigenvector is $\mathbf{1}_n$, and the rest of eigenvalues have positive real part.
\end{lemma}
\begin{proof}
	We look at $\kappa\Lambda$ as a perturbation matrix of the Laplacian matrix $L = BD_{\omega}B^T$, which is positive semidefinite with a single zero eigenvalue whose eigenvector is $\mathbf{1}_n$ if the graph of the corresponding framework is connected. We then note that
\begin{equation}
	L-\kappa\Lambda = (BD_\omega-\kappa M)B^T.
\end{equation}
	Therefore, the eigenvalue zero and its corresponding eigenvector $\mathbf{1}_n$ coming from $L$ are not perturbed for a generic $\kappa$ with an arbitrary $M$ since $B^T\mathbf{1}_n = 0$. The rest of eigenvalues are continuous functions of $\kappa$, hence, the rest of eigenvalues of $(L-\kappa\Lambda)$ are arbitrarily close to the ones of $L$ (all of them with positive real part) for a sufficiently small $|\kappa|$.
\end{proof}

We will discuss how small should be $|\kappa|$ in the Subsection \ref{sec: kappa}. As in the standard consensus algorithm in the Euclidean space, for the sake of clarity, and without loss of generality, let us focus only on one arbitrary component of the involved vectors in the analysis of the following main result.

\begin{theorem}
\label{thm: movdis}
Consider a one-dimensional, i.e., $m=1$, reference shape $p^*\in\mathbb{R}^{n}$ for the desired shape $\mathcal{S}$ in (\ref{eq: dshape}) and the framework $\mathcal{F} = (\mathcal{G},p)$, whose graph is connected. Consider the following distributed control action for the dynamics (\ref{eq: dyn}):
\begin{equation}
	u_i = - \sum_{j\in\mathcal{N}_i}\Big( \tilde\omega_{ij} (p_i - p_j) - \omega_{ij}(p^*_i - p^*_j)\Big), \, \forall i\in\mathcal{V},
	\label{eq: uwtil}
\end{equation}
	where $\omega_{ij}=\omega_{ji} > 0$, and the modified $\tilde\omega_{ij}$'s are designed following (\ref{eq: wmod}) and (\ref{eq: vs}) such that $v^*_i = v^*\in\mathbb{R}, \forall i\in\mathcal{V}$. If in the modified weights the constant $|\kappa|$ is sufficiently small, then the agents converge to the desired shape with steady-state velocity $v^*$, i.e.,  $p(t) \to \mathcal{S}$ and $\frac{\mathrm{d} p_i(t)}{\mathrm{dt}} \to v^*$ as $t\to\infty$.
\end{theorem}
\begin{proof}
First we note that (\ref{eq: uwtil}) is of the form of (\ref{eq: fij}), therefore the agent $i$ only needs local information to implement the control action (\ref{eq: uwtil}). We start by plugging (\ref{eq: uwtil}) into (\ref{eq: dyn}) for all the agents to arrive at the following compact form
\begin{equation}
	\dot p = -(L-\kappa\Lambda)p + Lp^*,
	\label{eq: Lwtil}
\end{equation}
which can be trivially rewritten as
\begin{equation}
	\dot p  + (L-\kappa\Lambda)p = Lp^*.
	\label{eq: Lwtil2}
\end{equation}
The solution to the ordinary differential equation (\ref{eq: Lwtil2}) is the solution of its homogeneous part plus a particular solution. It is well known that the solution $p_h(t)$ of the homogeneous equation
\begin{equation}
\dot p  + (L-\kappa\Lambda)p = 0,
\end{equation}
is given by
\begin{equation}
	p(t) = \operatorname{exp}\{-(L-\kappa\Lambda)t\}p(0),
\end{equation}
where the exponential can be calculated from the Jordan form of $(L-\kappa\Lambda)$. In particular, the solution is given by
\begin{equation}
	p_h(t) = c_1\mathbf{1}_n + \sum_{l=2}^n f_l(t,c_l,\dots,c_n,w_l,\dots,w_n)e^{-\lambda_lt},
	\label{eq: ph}
\end{equation}
	where $\lambda_l\in\mathbb{C}$ and $w_l\in\mathbb{R}^n$ are eigenvalues and (possibly generalized) eigenvectors of the modified Laplacian matrix respectively, $c_l\in\mathbb{R}$ are the constants to be calculated from the initial condition $p(0)$, and the functions $f_l(t,c_l,\dots,c_n,w_l,\dots,w_n)$ correspond to linear combinations like $c_lw_l + c_{l+1}(w_{l+1}+w_lt)+\dots$ depending on the algebraic and geometric multiplicity of the corresponding eigenvalues of the modified Laplacian matrix. Indeed, for the first term of (\ref{eq: ph}), according to Lemma \ref{lem: lmod}, we have that $\lambda_1 = 0$ with $w_1 = \mathbf{1}_n$, and for the rest of terms we have that $\lambda_l > 0, 2\leq l\leq n$ if $|\kappa|$ is sufficiently small. Therefore, $p_h(t) \to c_1\mathbf{1}_n$ as $t\to\infty$ as in the standard consensus algorithm.

Now we are going to verify that
	\begin{equation}
		p_p(t) = (\kappa\Lambda p^*) t + p^*,
	\label{eq: pp}
	\end{equation}
is a particular solution of (\ref{eq: Lwtil2}). First, we plug (\ref{eq: pp}) into (\ref{eq: Lwtil2})
	\begin{align}
		\kappa\Lambda p^*+(\kappa L\Lambda p^*)t+Lp^*-(\kappa^2\Lambda^2p^*)t-\kappa\Lambda p^*&=Lp^* \nonumber \\
		(\kappa L\Lambda p^* - \kappa^2\Lambda^2p^*) t &= 0.
		\label{eq: pp2}
	\end{align}
Now we check that according to our design in (\ref{eq: vs}), or (\ref{eq: vs2}) for the compact form, we have that
\begin{equation}
	\kappa^2\Lambda^2 p^* = \kappa^2 MB^T\Lambda p^* = \kappa MB^T (\mathbf{1}_n\otimes v^*) = 0,
\end{equation}
and together with $\Lambda p^* \in \operatorname{Ker}\{L\}$, we can conclude that (\ref{eq: pp2}) is true for all $t$. Hence, (\ref{eq: pp}) is a particular solution of (\ref{eq: Lwtil2}). Consequently, the solution of (\ref{eq: Lwtil2}) is the following combination of (\ref{eq: ph}) and (\ref{eq: pp})
\begin{align}
	p(t) &= p_h(t) + p_p(t) \nonumber \\
	&= c_1\mathbf{1}_n +(\kappa\Lambda p^*) t+ p^* + \sum_{l=2}^n f_l e^{-\lambda_l t}.
	\label{eq: sol}
\end{align}
Again, according to Lemma \ref{lem: lmod}, for a sufficiently small $|\kappa|$ we have that $\lambda_l > 0, 2 \leq l \leq n$ in (\ref{eq: sol}), therefore we can conclude that
\begin{equation}
	p(t) \to c_1\mathbf{1}_n + p^* +(\kappa\Lambda p^*) t \in\mathcal{S}, \quad t\to\infty,
\end{equation}
i.e., the reference shape $p^*$ will move in a translational motion following the constant velocity $\kappa\Lambda p^* = (\mathbf{1}_n\otimes v^*)$, and $c_1$ will depend on the initial condition $p(0)$.
\end{proof}

We remind that in the upcoming Subsection \ref{sec: kappa}, we will see how small should be $|\kappa|$ to apply the Theorem \ref{thm: movdis}.

The Theorem \ref{thm: movdis} is a generalization of the displacement-consensus-based formation control algorithm debriefed in \cite{Oh2015}. Indeed, for $\Lambda = 0$, i.e., for all $\mu_{ij} = 0$, then the particular solution (\ref{eq: pp}) is $p_p(t) = p^*$, i.e., a static desired shape. We remind that Theorem \ref{thm: movdis} was dealing with single components, e.g., the \emph{x-components}, of the involved vectors. For the general $m$-dimensional case with $p^*\in\mathbb{R}^m$, from Theorem \ref{thm: movdis} we can derive straightforwardly that $p(t) \to \sum_{l=1}^m \left(c_l(\mathbf{1}_n \otimes \mathbf{\hat 1}_l)\right) + p^* + (\kappa \overline\Lambda p^*)t\in\mathcal{S}$ as $t\to\infty$, where the different $c_l\in\mathbb{R}$ depend on the initial condition $p(0)\in\mathbb{R}^m$.

\subsection{Comparison with other maneuvering techniques}
A fair observation might remark that the straightforward protocol $u = -\overline L (p-p^*) + \mathbf{1}_n \otimes v^*$ also solves the considered problem of this section in an arguably more natural way. Such a straightforward protocol sets $v^*$ in open-loop as an independent term from the formation. In our protocol (\ref{eq: Lwtil}), the term responsible for the eventual velocity of the formation is $\kappa \Lambda p(t) = \kappa Mz(t)$. This term depends on the relative positions of the formation, which are in closed-loop. For example, the speed and heading of the \emph{formation's velocity} react accordingly to the current scale and orientation of the shape of the formation, and eventually, we have that $\kappa Mz(t)\to v^*$ as $t\to\infty$ since $z(t)\to z^*$ as $t\to\infty$. This reactive property allows the designer to program complex reactive motion behaviors for the formation since they will depend on $z(t)$, e.g., by transitioning between different and possibly time-varying $z^*$. There is also the possibility of extending the proposed technique by allowing time-varying $\mu_{ij}(t)$ in order to program desired time-varying $v^*(t)$ without forgetting that the \emph{formation's velocity} is always coupled with $z(t)$ (and eventually with a possible time-varying $z^*(t)$). This extension can be studied using output regulation techniques such as the \emph{internal model principle} \cite{Marina15} with a network of \emph{leaders-followers}, even possibly allowing communication between agents for the recalculation of new $\mu_{ij}$ if needed.

\subsection{General design for the motion parameters $\mu_{ij}$}
\label{sec: musdis}
The design of an arbitrary $v^*$ by following (\ref{eq: vs}) needs from each agent to have at least $m$ independent $(p_i^* - p_j^*)$ vectors. If this requirement is not met, then in a second approach, we can still construct an arbitrary $v^*$ by employing $\mu_{ij}\in\mathbb{R}^{m\times m}$. For example, for the agent $4$ in Figure \ref{fig: sqdis} we can set $\mu_{41} = \left[\begin{smallmatrix}a & -b \\ b & a\end{smallmatrix}\right], a,b\in\mathbb{R}$, i.e.,
\begin{equation}
v^*_4 = \kappa \left(a\begin{bmatrix}1 & 0 \\ 0 & 1\end{bmatrix}(p_4^*-p_1^*) + b\begin{bmatrix}0 & -1 \\ 1 & 0\end{bmatrix}(p_4^*-p_1^*)\right).
\end{equation}
With this more general approach where $\mu_{ij}$ are matrices, we cannot modify the weights of $L$ anymore as in (\ref{eq: wmod}). In fact, we will modify directly $\overline L$ for the $m$-dimensional case. Let us define
\begin{equation}
\hat\Lambda := \hat M \overline B^T,
	\label{eq: Lamh}
\end{equation}
where $\hat M \in\mathbb{R}^{m|\mathcal{V}|\times m|\mathcal{E}|}$ consists of the following $m \times m$ blocks
\begin{equation}
	\hat M = \begin{bmatrix}\mu_{11} & \cdots & \mu_{1|\mathcal{E}|} \\ \vdots & \ddots & \vdots \\ \mu_{|\mathcal{V}|1} & \cdots & \mu_{|\mathcal{V}||\mathcal{E}|} \end{bmatrix}.
\label{eq: Mhat}
\end{equation}
Note that for the particular case where $\mu_{ij}\in\mathbb{R}$ as in (\ref{eq: M}), then $\hat M = \overline M = M \otimes I_m$, and consequently $\hat\Lambda = \overline\Lambda = \overline{MB^T}$. With these more general definitions on hand, we can generalize the control action (\ref{eq: uwtil}) as
\begin{equation}
	u_i = - \sum_{j\in\mathcal{N}_i} \left(\omega_{ij} \Big( (p_i - p_j) - (p^*_i - p^*_j)\Big) + \kappa \mu_{ij} (p_i - p_j)\right),
	\label{eq: udismu}
\end{equation}
with the corresponding compact form considering all the agents
\begin{equation}
	\dot p = -(\overline L-\kappa\hat\Lambda)p + \overline Lp^*.
	\label{eq: Lwtil3}
\end{equation}
In general, the dynamics of the different Euclidean components of $p(t)$ in (\ref{eq: Lwtil3}) are not decoupled anymore because of (\ref{eq: Mhat}); therefore, we cannot analyze them separately. Nevertheless, the extension of Theorem \ref{thm: movdis} for the closed loop (\ref{eq: Lwtil3}) is still straightforward since the homogeneous part of (\ref{eq: Lwtil3}) is still a \emph{perturbed} consensus algorithm, and its particular solution is $p_p(t) = (\kappa\hat\Lambda p^*)t + p^*$, which is a mere extension of (\ref{eq: pp}). Note that because of the definition of $\hat\Lambda$ in (\ref{eq: Lamh}), the extension of Lemma \ref{lem: lmod} for $\hat\Lambda$ is straightforward as well.

\subsection{How small should be $|\kappa|$?}
\label{sec: kappa}
In general, the requirement of a sufficiently small $|\kappa|$ in Lemma \ref{lem: lmod} is not a conservative condition.
If the graph $\mathcal{G}$ in the framework $\mathcal{F}$ does not contain any \emph{cycles}, i.e., it is a \emph{tree graph}, then we can exploit the fact that $B^TB$ is a positive definite matrix \cite{dimarogonas2008stability} to establish a bound for $|\kappa|$ in Theorem \ref{thm: movdis}. In the context of formation control, such a graph means that no relative position is a function of the others. For the sake of simplicity and without loss of generality, let us set $\omega_k = \omega^*, k\in\{1,\dots,|\mathcal{E}|\}$ in the following result, that is meaningful when the eventual desired $v^* \neq 0$.
\begin{proposition}
	\label{prop: kappa}
	Consider $\omega_k = \omega^* > 0, k\in\{1,\dots,|\mathcal{E}|\}$ in (\ref{eq: Lwk}). If $\mathcal{G}$ is a tree graph, and 
\begin{equation}
	|\kappa| < \omega^* \frac{\lambda_{\text{min}}(\overline{B^TB})}{||\overline B^T\hat M||_2},
\end{equation}
	where $||X||_2$ and $\lambda_{\text{min}}(X)$  denote the spectral norm and  the smallest eigenvalue of $X$ respectively, then $|\kappa|$ is sufficiently small as required in Lemma \ref{lem: lmod} and Theorem \ref{thm: movdis}.
\end{proposition}
\begin{proof}
	Let us define the error signal $e(t) := z(t) - z^*$, where $z$ is as in (\ref{eq: z}) and $z^* = \overline B^Tp^*$. We are going to find a bound for $|\kappa|$ such that we can guarantee $e(t)\to 0$ as $t\to\infty$. In order to calculate the dynamics of the error signal $e$ we first accommodate (\ref{eq: Lwtil3}) as follows
\begin{align}
	\dot p &= -\overline{BD_\omega B^T}(p-p^*) + \kappa\hat \Lambda p \nonumber \\
	&= -\omega^*\overline {B} e + \kappa \hat M z \nonumber \\
	&= -\omega^*\overline {B} e + \kappa \hat M e + \kappa\hat  M z^* \nonumber \\
	&= -\omega^*\overline {B} e + \kappa \hat M e + \kappa (\mathbf{1}_n\otimes v^*) \label{eq: ppro}
\end{align}
and knowing that $\dot e = \dot z = \overline B^T \dot p$, we have that
\begin{align}
\dot e &= -\omega^*\overline B^T \overline {B} e + \kappa \overline B^T\hat M e + \kappa \overline B^T (\mathbf{1}_n\otimes v^*) \nonumber \\
&= -\omega^*\overline B^T \overline {B} e + \kappa \overline B^T\hat M e.
\end{align}
Consider the following Lyapunov candidate function $V = \frac{1}{2}||e||^2$, then we have that
\begin{align}
	\frac{\mathrm{d}V}{\mathrm{dt}} &= - \omega^* e^T \overline{B^TB}e + \kappa e^T\overline B^T\hat M e \nonumber \\
	&\leq -\left(\omega^*\lambda_{\text{min}}(\overline{B^TB}) - |\kappa|\,|| \overline B^T\hat M||_2\right)||e||^2,
\end{align}
	therefore if $|\kappa| < \omega^* \frac{\lambda_{\text{min}}(\overline{B^TB})}{||\overline B^T\hat M||_2}$, then $\frac{\mathrm{d}V}{\mathrm{dt}} \leq -\alpha V$ for some $\alpha > 0$. Note that $|\kappa|$ can be bigger than zero since $\lambda_{\text{min}}(\overline{B^TB}) > 0$ because $\mathcal{G}$ is a tree graph. Since $e(t) \to 0$ exponentially fast as $t\to\infty$, we have that in (\ref{eq: ppro}) $\dot p(t) \to \kappa (\mathbf{1}_n\otimes v^*)$ exponentially fast as well as $t\to\infty$. Therefore, the configuration $p(t)$ converges to the desired shape $\mathcal{S}$ following the translational velocity $\kappa(\mathbf{1}_n\otimes v^*)$.
\end{proof}
For general graphs containing cycles, we would need to find a coordinate transformation for $z$ such that we separate the independent coordinates from the dependent ones. For example, in a \emph{triangular formation} with $\mathcal{E} = \{(1,2),(2,3),(3,1)\}$ we have three relative positions with their respective dynamics. However, they are not independent since we have the constraint $z_{12} + z_{23} + z_{31} = 0$. According to the definition of the error signal $e(t)$ in Proposition \ref{prop: kappa} we have that $(z-z^*)^T \overline B^T\overline B(z-z^*) = 0$ for some non-zero $(z-z^*)$. If we find the Jordan form $J$ of $\overline B^T\overline B$ such that $\overline B^T\overline B = T\left[\begin{smallmatrix}J_1 & 0 \\ 0 & J_2\end{smallmatrix}\right]T^{-1}$ with $J_2$ being the Jordan block with all zeros in its diagonal, then we can apply the coordinate transformation $Tz = \left[\begin{smallmatrix}z_{\text{indep}} \\ z_{\text{dep}} \end{smallmatrix}\right]$ so that we can look for an admissible $\kappa$ by carrying out a stability analysis as in Proposition \ref{prop: kappa} but focusing on the error signal of the independent coordinates $(z_{\text{indep}}(t)- z_{\text{indep}}^*)$.

\section{Robustness issues due to imperfect sensing}
Following the spirit in \cite{mou2016undirected,Marina15,meng2016formation}, this section uncovers that an eventual (non-desired) motion, together with a \emph{distorted} shape, of the formation can be triggered by a disagreement over the measurement of the relative position between two neighboring agents $i$ and $j$. Being more specific, an agent might not measure a relative position correctly, e.g.,
\begin{align}
	(p_j-p_i)|_{\text{measurement}} &= aR(p_j -p_i)|_{\text{actual relative position}}, \nonumber \\ a&\in\mathbb{R}_+, R\in SO(m).
	\label{eq: sfactor}
\end{align}
The relation (\ref{eq: sfactor}) might be given when a robot measures such a relative position with two sensors, namely, range and direction, i.e., $z_{ij} = ||z_{ij}||\frac{z_{ij}}{||z_{ij}||}$. We then consider that the range sensor has a different scale factor than $1$, and the direction sensor is biased by a constant rotation matrix $R\neq I_m$, e.g., it can be seen as a misaligned compass with respect to \emph{North} in 2D or $m=2$. The disagreement between two neighboring agents shows up when their sensors have a different scale factor and/or misalignment, e.g., $a_i\neq a_j$ and $R_i\neq R_j$. 

Let us illustrate the robustness issue with the simple example in $m$-dimensions where the graph of the framework is given by $\mathcal{N} = \{1,2\}$ and $\mathcal{E} = \{(1,2)\}$. In such a case, the agents implement the control action derived from (\ref{eq: uLcon}) by considering $\omega_1 = 1$
\begin{equation}
	\begin{cases}
	\dot p_1 &= -(z_{12} - z_{12}^*) \\
	\dot p_2 &= -(z_{21} - z_{21}^*).
	\end{cases}
	\label{eq: dynsf}
\end{equation}
However, the agents execute (\ref{eq: dynsf}) with their onboard measurements. Consider that the agent $1$ measures correctly the relative position $(p_i-p_j)$, but agent $2$ measures it with an arbitrary scale factor as in (\ref{eq: sfactor}). Assume further that agents $1$ and $2$ have the same target for their relative positions, i.e., $z_{12}^* = -z_{21}^*$. Since the authors in \cite{meng2016formation} have covered the \emph{2-agent} case in 2D for $R\neq I_2$ (and implicitly $a = 1$) in (\ref{eq: sfactor}), let us focus solely on the two agents with different scale factors, e.g., $1$ and $a\in\mathbb{R}_+\setminus\{1\}$, but same alignments $I_m$. Then, (\ref{eq: dynsf}) will be executed by the agents as
\begin{equation}
	\begin{cases}
	\dot p_1 &= -(I_mz_{12} - z_{12}^*) \\
		\dot p_2 &= aI_mz_{12} - z_{12}^*,
	\end{cases}
	\label{eq: exscale}
\end{equation}
and with a bit of algebraic manipulation we arrive at
\begin{equation}
	\begin{cases}
		\dot p_1 &= -\big((b+1)z_{12} - z_{12}^*\big) + bz_{12}\\
		\dot p_2 &= \big((a-c)z_{12} - z_{12}^*\big) + cz_{12},
	\end{cases}
	\label{eq: dynsf2}
\end{equation}
where the values of $b,c\in\mathbb{R}$ will be found out shortly. By inspecting the first terms of (\ref{eq: dynsf2}), we can deduce that the formation will achieve a \emph{distorted} steady-state shape satisfying
\begin{equation}
\tilde z_{12} = \frac{z_{12}^*}{b+1} = \frac{z_{12}^*}{a-c},
\label{eq: shsf}
\end{equation}
therefore, we have that $b + 1 = a - c$. If we consider the second terms in (\ref{eq: dynsf2}) responsible for a common \emph{residual} velocity, then the second equation to find out the values of $b$ and $c$ is given by the non-desired eventual velocity of the formation when the two agents are at the \emph{distorted} steady-state relative position (\ref{eq: shsf}), i.e., when $\lim_{t\to\infty}\dot p_1(t) = \lim_{t\to\infty}\dot p_2(t) = b\frac{z_{12}^*}{b+1} = c\frac{z_{12}^*}{a-c}$. Then we arrive at $b = c = \frac{a-1}{2}$. This example reveals that a different scale factor between two agents not only induce an \emph{expected} distortion in the desired shape but an \emph{unexpected} translational motion, whose speed depends on how far from one is the scale factor in agent $2$, i.e., $\lim_{t\to\infty}\dot p_1(t) = \lim_{t\to\infty}\dot p_2(t) = \frac{a-1}{a+1}z^*_{12}$.

For general frameworks, it is a matter of finding out the particular relations between the different scale factors and misalignments of the agents and the eventual distorted shape so that we can find out a matrix (to be introduced later) playing a similar role as $\hat M$ in (\ref{eq: Mhat}) to determine whether there is a residual steady-state motion. 

Let $a = \begin{bmatrix}a_1 \dots a_n\end{bmatrix}^T$ be the stacked vector of scale factors $a_i \in\mathbb{R}^+$ for each agent, and let $R = \begin{bmatrix}R_1^T \dots R_n^T\end{bmatrix}^T$ be the stacked matrix of rotational matrices $R_i\in SO(m)$ representing the misalignment for each agent. Finally, let us define $D_x := \overline D_a D_R$, and note that $D_x$ is always invertible since $a_i > 0, i\in\mathcal{V}$. Then, by following the example (\ref{eq: exscale}), we  add the scale factors and the misalignments of each agent to the closed loop derived from (\ref{eq: Lcon}) as
\begin{equation}
\dot p = - D_x\overline L p + \overline Lp^*.
	\label{eq: cldx}
\end{equation}
Note that the dynamics (\ref{eq: cldx}) that consider imperfect measurements are substantially different than (\ref{eq: Lwtil3}) with the design of a desired translational motion. In particular, the \emph{modified} Laplacian matrix for the $m$-dimensional case is now $D_x\overline L$ and not just the \emph{linear modification} $(\overline L-\kappa\hat\Lambda)$. Nevertheless, they both share $0$ as a single eigenvalue whose eigenvector is $\mathbf{1}_{mn}$, and a statement similar to Lemma \ref{lem: lmod} can be made if $D_x\approx I_{mn}$ (instead of $|\kappa|$ sufficiently small). Similarly as in (\ref{eq: dynsf2}), let us now rewrite (\ref{eq: cldx}) as
\begin{align}
	\dot p &= - D_x\overline L p + \overline Lp^* + \breve M\overline B^Tp -  \breve M\overline B^Tp \nonumber \\
	&= - D_x(\overline B\overline D_\omega \overline B^T)p + \overline Lp^* + \breve M\overline B^Tp - \breve M\overline B^Tp \nonumber \\
	&= -(D_x\overline B\overline D_\omega + \breve M)\overline B^Tp + \overline Lp^* + \breve M\overline B^Tp,
	\label{eq: compactsf}
\end{align}
where $\breve M$ has the same dimensions as $\hat M$ in (\ref{eq: Mhat}). However, the elements of $\breve M$ are not constructed from any $\mu_{ij}$. The matrix $\breve M$ plays the role of assisting us in understanding and calculating the \emph{residual} steady-state velocity emerging from wrong measurements. The simulations will indicate that all the elements of $\breve M$ are different from zero in general. Physically, this fact reveals that all the relative positions of the framework contribute to the steady-state velocity of the agent $i$. Let us now introduce and define formally the distorted shape $\tilde{p}^*\in\mathbb{R}^{mn}$, and the residual steady-state velocity $\tilde v^*\in\mathbb{R}^m$. Similarly as in (\ref{eq: dynsf2}), we will find out the values for $\tilde{p}^*$, $\tilde v^*$ and $\breve M$ such that they satisfy the two coupled conditions extracted from (\ref{eq: compactsf})
\begin{equation}
	\begin{cases}
		(D_x\overline B\overline D_\omega + \breve M)\overline B^T \tilde p^* =  \overline L p^*\\
	\breve M\overline B^T \tilde p^* = \mathbf{1}_n\otimes \tilde v^*,
	\end{cases}
	\label{eq: condsfcases}
\end{equation}
that can be combined in the following single condition
\begin{equation}
	D_x \overline L\tilde p^* - \overline Lp^* = -(\mathbf{1}_n\otimes \tilde v^*).
\label{eq: condsf}
\end{equation}
Note that the first condition in (\ref{eq: condsfcases}) is related to the (static) equilibrium $\tilde p^*$ of a \emph{distorted}-displacement-based formation controller. We identify such controller from the first two terms in (\ref{eq: compactsf}), i.e,  $\dot p = -(D_x\overline B\overline D_\omega + \breve M)\overline B^Tp + \overline Lp^*$. Similarly as in Theorem \ref{thm: movdis}, if we are at the equilibrium $\tilde p^*$, then the first two terms in (\ref{eq: compactsf}) vanish and the third one is the responsible for the residual motion of $\tilde p^*$.

We have a trivial case for (\ref{eq: condsf}) when $a_i = a^* \in\mathbb{R}^+, R_i=R^*, i\in\mathcal{V}$. Noting first that $(I_n \otimes R^*)(L \otimes I_m) = (L\otimes R^*) = (L \otimes I_m)(I_n \otimes R^*)$, then we have that
\begin{equation}
	\overline L (a^* (I_n\otimes R^*) \tilde p^* - p^*) = -(\mathbf{1}_n\otimes \tilde v^*),
	\label{eq: condsf2}
\end{equation}
which can be satisfied if and only if $\tilde v^* = 0$, and $\tilde p^* = \frac{1}{a^*}(I_n\otimes {R^*}^T)p^* + (\mathbf{1}_n\otimes b)$, with $b\in\mathbb{R}^m$ being an arbitrary (offset) vector. Note that this is the only solution to the trivial case since $\mathbf{1}_n \notin \operatorname{Im}\{L\}$, therefore we have to make the bracket in the left hand side of (\ref{eq: condsf2}) to be in the kernel of $\overline L$. If all the agents have the same (wrong) perception about the relative positions, e.g., all the agents share the same scale factor $a^*$ and misalignment $R^*$ for measuring $z_{ij}$, then they will achieve a distorted but eventually static shape $\tilde p^* = \frac{1}{a^*}(I_n\otimes {R^*}^T)p^*$ if $R^*\approx I_m$, so we do not perturb much the positive eigenvalues of $\overline L$ in (\ref{eq: cldx}) as we will see. Let us now introduce the following more general result.
\begin{theorem}
\label{thm: 2}
Consider a desired shape constructed from $p^*$ and a framework $\mathcal{F}$ with a connected graph $\mathcal{G}$ without any cycles, and also consider the control action (\ref{eq: Lcon}) for the dynamics (\ref{eq: pdyn}). Consider that the sensing of the agents is close to perfect, i.e., $D_x \approx I_{mn}$. If at least one agent $l\in\mathcal{V}$ has a different scale factor and/or misalignment among all $a_i\in\mathbb{R}^+, R_i\in SO(m), i\in\mathcal{V}\setminus l$ for measuring its available relative positions $z_{lj}, j\in\mathcal{N}_l$ as in (\ref{eq: sfactor}), then the framework $\mathcal{F}$ will display a steady-state distorted shape $\tilde p^*$ (close to $p^*$) travelling with a residual steady-state (in general non-zero) velocity $(\mathbf{1}_n \otimes \tilde v^*)$.
\end{theorem}
\begin{proof}
We first calculate the steady-state \emph{distorted} relative position vectors resulting from (\ref{eq: cldx}), i.e, $\tilde z^* = \overline B^T \tilde p^*$. We start with multiplying by $\overline B^T$ both sides of (\ref{eq: condsf})
\begin{align}
	\overline B^T D_x\overline B\overline D_\omega \tilde z^* - \overline B^T\overline B \overline D_\omega z^* = 0 \nonumber \\
	\tilde z^* = (\overline B^T D_x\overline B\overline D_\omega)^{-1}\,\, \overline B^T\overline B \overline D_\omega z^*, \label{eq: ztilde}
\end{align}
	since $\mathcal{G}$ does not contain any cycles, then the inverse matrix in (\ref{eq: ztilde}) exists. Note that if $D_x = I_{mn}$ then $\tilde z^* = z^*$. The residual velocity $\tilde v^*$ for the formation is calculated then from (\ref{eq: condsf}) as
\begin{align}
&\Big(D_x\overline B\overline D_\omega (\overline B^T D_x\overline B\overline D_\omega)^{-1}\,\, \overline B^T\overline B \overline D_\omega - \overline B\overline D_\omega \Big) z^* =  \nonumber \\
&\Big(D_x\overline B (\overline B^T D_x\overline B)^{-1}\,\, \overline B^T\overline B - \overline B \Big) \overline D_\omega z^* = -(\mathbf{1}_n \otimes \tilde v^*).
	\label{eq: vtilde}
\end{align}
Before the calculation of $\breve M$ in (\ref{eq: condsfcases}) we need to check that
\begin{align}
	&(B \otimes I_m)^T(I_n\otimes R^*)(B \otimes I_m) = (B^T \otimes I_m)(B\otimes R^*) = \nonumber \\
	&=(B^TB \otimes R^*) = (B^TB \otimes I_m)(I_{|\mathcal{E}|}\otimes R^*).
	\label{eq: use}
\end{align}
Now we take (\ref{eq: ztilde}) and (\ref{eq: vtilde}) for the left-hand and right-hand sides respectively of the second condition in (\ref{eq: condsfcases}).
\begin{align}
&\breve M (\overline B^TD_x\overline B\overline D_\omega)^{-1}\overline B^T\overline B\overline D_\omega z^* = \nonumber \\ &= -\Big (D_x\overline B(\overline B^TD_x\overline B)^{-1}\overline B^T\overline B - \overline B \Big)\overline D_\omega z^*,
\end{align}
therefore we can deduce that
\begin{align}
	\breve M &= -\Big (D_x\overline B(\overline B^TD_x\overline B)^{-1}\overline B\overline B^T - \overline B \Big)(\overline B^T\overline B)^{-1}(\overline B^TD_x\overline B\overline D_\omega) \nonumber \\
	&= -\Big( D_x\overline B - \overline B(\overline B^T\overline B)^{-1} \overline B^TD_x\overline B\Big) \overline D_\omega. \label{eq: Mbreve}
\end{align}
	Now we are going to show that if $D_x = a^*(I_n\otimes R^*)$, then $\breve M = 0$. Then, focusing on (\ref{eq: Mbreve}) by exploiting (\ref{eq: use}) we have that
\begin{align}
	\breve M &=-\Big(a^*(I_n\otimes R^*)\overline B - a^*\overline B(\overline B^T\overline B)^{-1}(\overline B^T\overline B)(I_{|\mathcal{E}|}\otimes R^*)\Big)\overline D_\omega \nonumber \\
	&= -\Big(a^*(I_n\otimes R^*)\overline B - a^*(I_n\otimes R^*)\overline B \Big)\overline D_\omega = 0
	\label{eq: RA}
\end{align}
Therefore, if $D_x \approx a^*(I_n\otimes R^*)$ with at least one block diagonal element different than the rest (for example, coming from the agent $l\in\mathcal{V}$), then $\breve M$ is close to $0$. Similarly as in Theorem \ref{thm: movdis}, if $D_x \approx I_{mn}$ so that we do not perturb much the non-zero eigenvalues of $\overline L$, then we can check that the solution to the homogeneous part of (\ref{eq: cldx}) satisfies $p_h(t) \to \sum_{l=1}^m \left(c_l(\mathbf{1}_n \otimes \mathbf{\hat 1}_l)\right)$ as $t\to\infty$ with $c_l\in\mathbb{R}$ depending on the initial condition $p(0)$, and that a particular solution of (\ref{eq: cldx}) is $p_p(t) = \tilde p^* + (\breve M \overline B^T \tilde p^*) t = \tilde p^* + (\mathbf{1}_n \otimes \tilde v^*)t$. Note that we do not need explicitly $\tilde p^*$ but $\overline B^T\tilde p^* = \tilde z^*$ in (\ref{eq: cldx}) to test the particular solution. Nevertheless, similarly as in (\ref{eq: pstar}), the distorted configuration $\tilde p^*$ can be obtained from placing the origin $O_g$ at the center of masses $O_b$ of the shape described by the calculated $\tilde z^*$ in (\ref{eq: ztilde}).
\end{proof}
Note that the residual steady-state velocity from (\ref{eq: vtilde}) satisfies $\breve M\tilde z^* = (\mathbf{1}\otimes \tilde v^*)$ which in our simulations it seems to be different from zero if $D_x \approx a^*(I_n \otimes R^*)$ with different diagonal blocks. However, we have not proven that $\tilde z^* \notin \operatorname{Ker}\{\breve M\}$. This leaves an open problem on how to design $\mathcal{G}$, $\omega$, and other parameters in the framework such that the resultant formation is more robust against disagreements in sensing among the agents. Recently, the authors in \cite{ahn2019consensus} analyzed the stability condition $D_x\approx I_{2n}$ with scale factor $a=1$ for 2D formations.
\begin{remark}
According to (\ref{eq: sfactor}), agents have always the right perception about $0$ regardless of the scale factor and misalignment, i.e, the traditional consensus algorithm with $p^* = 0$ does not suffer from the uncovered robustness issues.
\end{remark}

One ad hoc solution to avoid this eventual undesired motion might be to fix one of the agents in the formation, e.g., set $\dot p_1 = 0$. However, displacement-based formation control might be required to work together with other algorithms, possibly involving the motion of the whole formation as illustrated in \cite{kia2018tutorial}. Note that this ad hoc solution will not prevent either the system to be unstable if one of the eigenvalues of the corresponding \emph{perturbed}/\emph{modified} Laplacian matrix is on the left-half plane. Therefore, fixing one agent will not be a definite solution to the uncovered problem in this paper.

\section{Numerical experiment}
\label{sec: simulations}
In this section we validate the results from Theorem \ref{thm: 2}. We choose the reference shape $p^*$ and graph $\mathcal{G}$ as in Figure \ref{fig: sqdis9}. We generate randomly the following vectors of scale factors $a=$
\scalebox{0.48}{
	$\begin{bmatrix} 0.96843302 & 1.00873027 &  0.9546316 &  1.04510691 &  1.02358278 &  0.95203593 &  1.04006459 &  0.96226732 &  0.98482596 \end{bmatrix}$
}
and misalignments (in radians) between $\pm 10$ degrees $r=$ 
\scalebox{0.48}{
	$\begin{bmatrix}-0.15850664 & -0.13158391 & -0.07226048 & -0.07021736 &  0.03995607 & -0.11761143 & -0.0692078 & 0.16551018 & 0.1331908
\end{bmatrix}$
}
so that we take each element of $r$ and construct a 2D rotational matrix. From the results in Theorem \ref{thm: 2} we predict the following distortion $\tilde z^*$ for the eventual relative positions, and we compare it with the desired one (we stack the $z_{ij}^T$ for better visualization)
\begin{equation}\left[
		\begin{smallmatrix}1.26245792 & -10.06235038 \\   0.91330118 &  -9.92409706 \\  -10.65052789 &  -0.06477112 \\  -2.10520906 & 10.3052228 \\ -0.86465511 &  10.27960119 \\ -11.3181448 & 0.60885649 \\ -0.33526666 & 9.43683404 \\ -1.0403239 & -10.29440935\end{smallmatrix}\right] \,
		\left[\begin{smallmatrix}0 & -10 \\  0 & -10 \\ -10 &  0 \\  0 &  10 \\ 0 & 10 \\ -10 & 0 \\ 0 &  10\\ 0 & -10 \end{smallmatrix}\right] \nonumber.
\end{equation}
In Figure \ref{fig: errors}, we show how the signal $||z(t) - \tilde z^*||$ converges to zero as predicted, and we can notice such a distortion in the eventual shape described by the agents. Finally, the predicted residual velocity is $\tilde v^* = \begin{bmatrix}0.33086245 & -0.18446561\end{bmatrix}^T$ units/sec that matches with the agents' trajectories (in Figure \ref{fig: errors}) and the signal $||\frac{\mathrm{d}p(t)}{\mathrm{dt}} - (\mathbf{1}_n \otimes \tilde v^*)||$ converging to zero. 

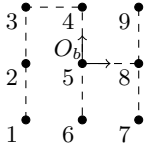
\begin{figure}
\centering
\begin{tikzpicture}[line join=round, scale=0.75]
\draw[dashed](-1,-1)--(-1,0)--(-1,1)--(0,1)--(0,0)--(0,-1);
\draw[dashed](0,0)--(1,0)--(1,-1);
\draw[dashed](1,0)--(1,1);
\filldraw(-1,-1) circle (2pt);
\filldraw(-1,0) circle (2pt);
\filldraw(-1,1) circle (2pt);
\filldraw(0,1) circle (2pt);
\filldraw(0,0) circle (2pt);
\filldraw(0,-1) circle (2pt);
\filldraw(1,-1) circle (2pt);
\filldraw(1,0) circle (2pt);
\filldraw(1,1) circle (2pt);
\draw[draw=black,arrows=->](0,0)--(0,0.5);
\draw[draw=black,arrows=->](0,0)--(0.5,0);
\node at (-0.25,0.25) {\small $O_b$ \normalsize};\node at (-.25,-.25) {\small $5$ \normalsize}; \node at (-1.25,-1.25) {\small $1$ \normalsize}; \node at (-1.25,-.25) {\small $2$ \normalsize}; \node at (-1.25,0.75) {\small $3$ \normalsize}; \node at (-.25,0.75) {\small $4$ \normalsize}; \node at (-.25,-1.25) {\small $6$ \normalsize}; \node at (0.75,-1.25) {\small $7$ \normalsize}; \node at (.75,-.25) {\small $8$ \normalsize}; \node at (.75,.75) {\small $9$ \normalsize};
\end{tikzpicture}
	\caption{Chosen reference shape $p^*$ for the numerical experiment. We choose a distance of $10$ units between agents in the $x$ and $y$ axes respectively. The ordered set of edges is $\mathcal{E} = \{(1,2),(2,3),(3,4),(4,5),(5,6),(5,8),(8,7),(8,9)\}$.}
	\label{fig: sqdis9}
\end{figure}

\begin{figure}
\centering
\includegraphics[width=0.49\columnwidth]{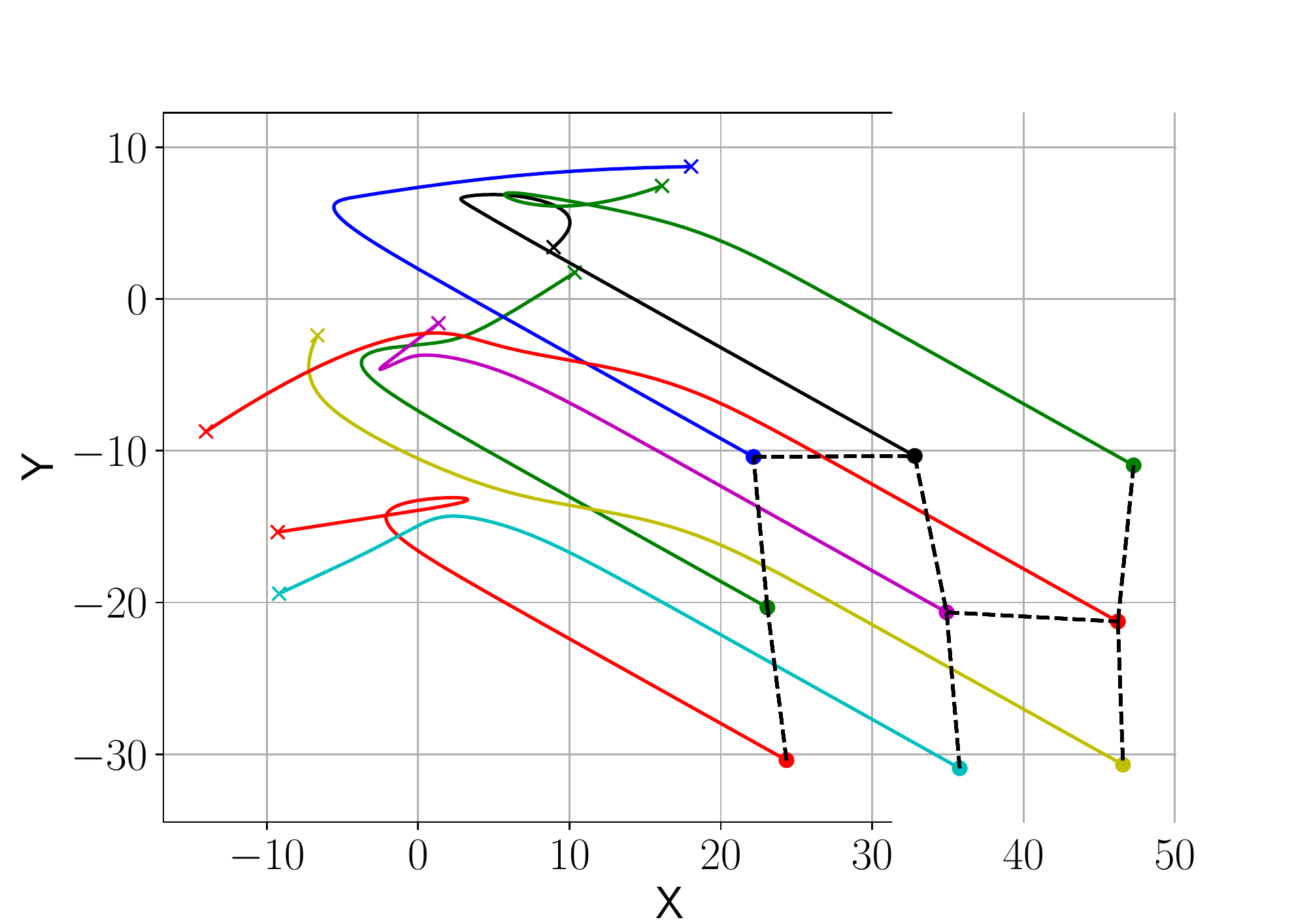}
\includegraphics[width=0.49\columnwidth]{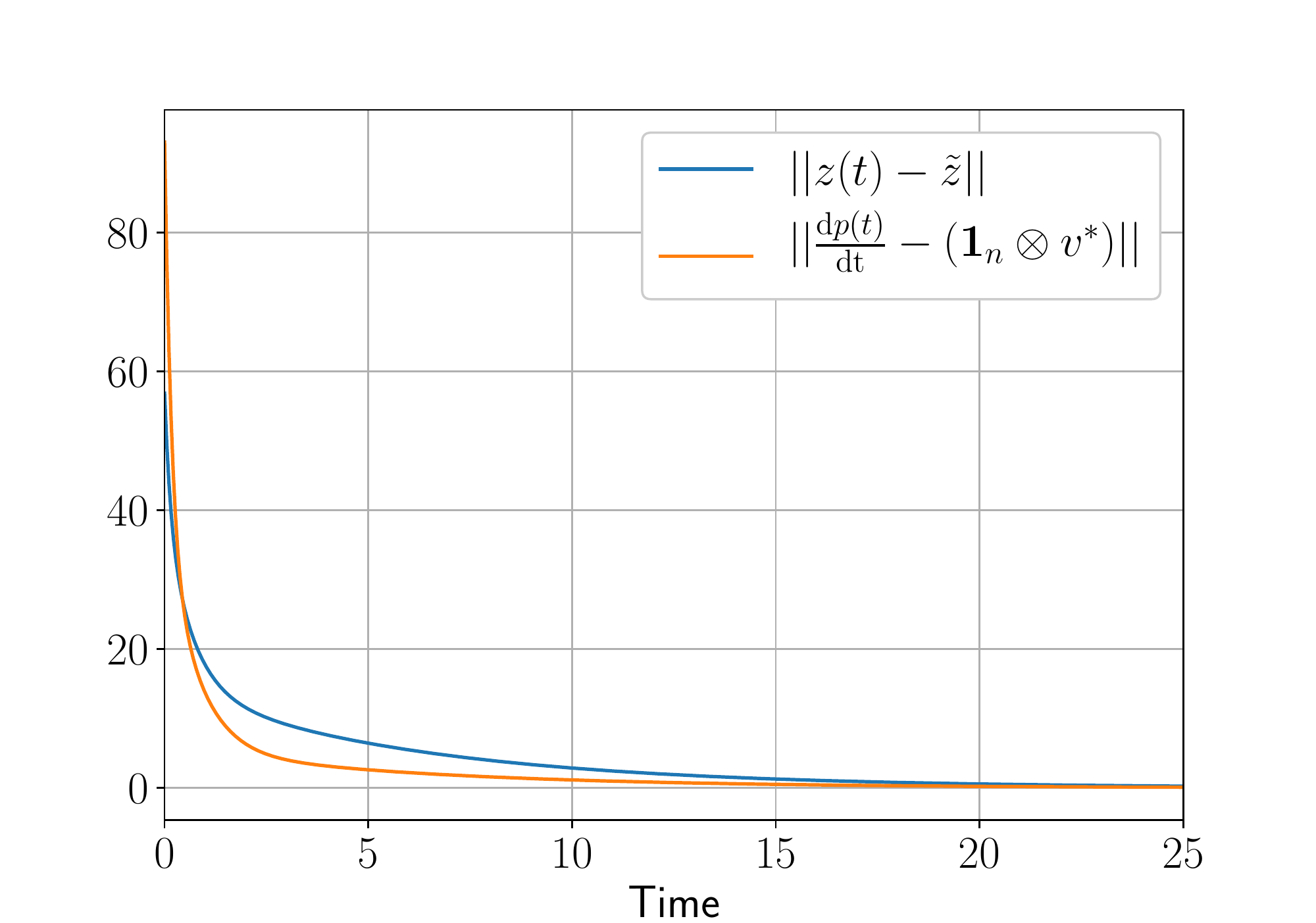}
	\caption{On the left side, we present the evolution of the agents' positions. The crosses and the solid circles denote for the $p(0)$ and $p(100)$ configurations respectively. Note the distortion of the shape at $p(100)$ with respect to the reference in Figure \ref{fig: sqdis9}, and that the formation does not stay static but travels with a constant translational velocity. On the right side, in blue color, the time evolution of the error norm between the relative positions in the framework and the predicted distorted relative positions. In orange color, the time evolution of the error norm between the velocities of the agents and the predicted residual translational velocity due to wrong measurements. }
\label{fig: errors}
\end{figure}

\section{Conclusions}
\label{sec: conc}
We have presented a technique to maneuver displacement-consensus-based formations by manipulating the weights of the standard Laplacian matrix. In fact, the standard displacement-consensus-based formation control is a particular case of our proposed technique. We uncovered that the displacement-consensus-based formation control has robustness issues when agents have different perceptions on measuring relative positions, namely, different scale factors and misalignments. These robustness issues manifest in the form of having the agents converging to a travelling \emph{distorted} shape. We have provided explicit expressions to calculate such distortion and undesired velocities in \emph{tree} formations in arbitrary $m$-dimensions.


\bibliographystyle{IEEEtran}
\bibliography{Bibs}

\end{document}